\documentclass[journal,twoside,web]{ieeecolor}
\usepackage{lcsys}

\makeatletter
\let\NAT@parse\undefined
\makeatother

\usepackage{graphicx}
\usepackage{amsmath,amssymb}
\usepackage{mathtools}
\usepackage{cite}
\usepackage{color}
\usepackage{booktabs,multirow}
\usepackage[flushleft]{threeparttable}
\usepackage{caption}
\newtheorem{assumption}{Assumption}
\newtheorem{definition}{Definition}
\newtheorem{theorem}{Theorem}
\newtheorem{lemma}{Lemma}
\newtheorem{corollary}{Corollary}
\newtheorem{remark}{Remark}

\usepackage{fancyhdr}
\usepackage{xcolor}

\usepackage{hyperref}
\hypersetup{
	colorlinks,
	linkcolor={red!50!black},
	citecolor={blue!50!black},
	urlcolor={blue!80!black}
}

\def\doi{10.1109/LCSYS.2022.3186236}
\fancyfootoffset{-1em}
\fancypagestyle{copyright}{
	\cfoot{\vspace{-1.8em}\footnotesize\textcopyright{} 2022 IEEE. Personal use of this material is permitted.  Permission from IEEE must be obtained for all other uses, in any current or future media, including reprinting/republishing this material for advertising or promotional purposes, creating new collective works, for resale or redistribution to servers or lists, or reuse of any copyrighted component of this work in other works.}
	\chead{\footnotesize This article has been accepted for publication in IEEE Control Systems Letters. This is the author's version which has not been fully edited and content may change prior to final publication. Citation information: \href{http://dx.doi.org/\doi}{DOI \doi}}
}

\begin{document}
\title{A simple suboptimal moving horizon estimation scheme with guaranteed robust stability}
\author{
	Julian D. Schiller, Boyang Wu, and Matthias A. Müller
	\thanks{This work was supported by the German Research Foundation (DFG) under Grant MU 3929/2-1. \textit{(Corresponding author: Julian D. Schiller.)}}
	\thanks{The authors are with the Leibniz University Hannover, Institute of Automatic Control, 30167 Hannover, Germany (e-mail: schiller@irt.uni-hannover.de; mueller@irt.uni-hannover.de).}
}

\maketitle
\thispagestyle{empty}
\thispagestyle{copyright}

\begin{abstract}
We propose a suboptimal moving horizon estimation (MHE) scheme for a general class of nonlinear systems.
To this end, we consider an MHE formulation that optimizes over the trajectory of a robustly stable observer.
Assuming that the observer admits a Lyapunov function, we show that this function is an $\boldsymbol{M}$-step Lyapunov function for suboptimal MHE.
The presented sufficient conditions can be easily verified in practice.
We illustrate the practicability of the proposed suboptimal MHE scheme with~a~standard nonlinear benchmark example.
Here, performing a single iteration is sufficient to significantly improve the observer's estimation results under valid theoretical guarantees.
\end{abstract}


\section{Introduction}
\IEEEPARstart{S}{tate} estimation is crucial for many control applications and hence of high practical relevance.
However, designing suitable estimators is a challenging problem for nonlinear systems, especially in case of noisy measurements and model inaccuracies.
To this end, moving horizon estimation (MHE) has proven to be a powerful solution and constitutes an active area of research with recent results providing sufficient conditions for robust stability, see, e.g., \cite{Rawlings2017,Schiller2022,Knuefer2021,Allan2021a}.
Since MHE is an optimization-based approach, it is usually computationally demanding; however, the computing power available in practice is often severely limited.
Accordingly, computing the global optimum at each time step is often not possible in practice, which, however, renders the theoretical guarantees invalid (since they crucially depend on this criterion).

\subsubsection*{Related Work}
To simplify the optimization problem, a pre-estimating MHE scheme for linear systems was proposed in \cite{Sui2010} that utilizes an additional \textit{auxiliary observer} to replace the system dynamics as MHE constraint.
Since no optimal disturbance sequence has to be computed, the number of optimization variables could be significantly reduced.
This idea was transferred to a class of nonlinear systems in~\cite{Suwantong2014}, and a major speed improvement compared to the standard MHE formulation could be shown.
However, this scheme relies on a uniform observability condition and requires optimal solutions to the MHE problem, the achievement of which can still hardly be guaranteed within fixed time intervals.
To overcome this, fast MHE methods were developed in, e.g.,~\cite{Kuehl2011,Wynn2014,Alessandri2017}, performing only a predetermined number of iterations of a certain optimization algorithm (e.g., gradient- or Newton-based).
Convergence guarantees could be established under observability conditions and using (local) contraction properties of the specific algorithms involved, cf. \cite{Wynn2014,Alessandri2017}.
In~\cite{Wan2017}, the combination of a fast MHE scheme and pre-estimation using a nonlinear Luenberger observer was considered.
A suboptimal proximity-MHE scheme for nonlinear systems was presented in \cite{Gharbi2021}, where nominal stability guarantees could be given without performing any optimization by employing a pre-stabilizing observer and a certain gradient-based optimization method.
In contrast, robust guarantees for suboptimal MHE were established in \cite{Schiller2022b} independent of the horizon length, the chosen optimization algorithm, and the number of iterations performed by using an observer-based candidate solution to the MHE problem.

\subsubsection*{Contribution}
We propose a simple suboptimal MHE scheme with least squares objective and exponential time-discounting. Such a discount factor has proven useful in recent works on nonlinear MHE, compare, e.g.,~\cite{Schiller2022,Knuefer2021,Schiller2022b}.
In contrast to those works, however, we optimize over the trajectory of a robustly stable auxiliary observer, similar to the idea proposed in \cite{Sui2010,Suwantong2014,Wan2017}, cf. Section~\ref{sec:mhe_formulation}.
Assuming that the observer admits a Lyapunov function, we show that this directly yields a novel $M$-step Lyapunov function for suboptimal MHE, independent of the chosen optimization algorithm and the number of iterations performed, cf.~Section~\ref{sec:MHE_stability}.
The stated sufficient condition on the horizon length~$M$ can be easily verified in practice.
We provide good tuning opportunities and consider modifications which allow for arbitrary horizon lengths.
Moreover, in contrast to \cite{Suwantong2014,Wan2017,Schiller2022b,Gharbi2021}, we show that the theoretical guarantees (both the decrease rate and disturbance gains) strictly improve with increasing horizon length and converge (for $M\rightarrow\infty$) to those of the auxiliary observer (which is the best possible bound given that we derive guarantees for an arbitrary number of iterations including zero).
We illustrate the applicability of the proposed suboptimal MHE scheme with a standard nonlinear MHE benchmark example, cf.~Section~\ref{sec:num}.
We verify the sufficient conditions and show that performing only a single iteration of the optimizer each time step is sufficient to significantly improve the estimation results from the auxiliary observer under valid theoretical~guarantees.

\subsubsection*{Notation}
The set of all integers in an interval $[a,b]\subset\mathbb{R}$ is denoted by $\mathbb{I}_{[a,b]}$ and the set of integers greater than or equal to $a$ by $\mathbb{I}_{\geq a}$.
Let $\|x\|$ denote the Euclidean norm of the vector $x \in \mathbb{R}^n$ and $\|x\|_A^2=x^\top A x$ for a positive definite matrix $A=A^\top$.
The minimal (maximal) eigenvalues of $A=A^\top$ are denoted by $\lambda_{\min}(A)$ ($\lambda_{\max}(A)$) and the maximum generalized eigenvalue w.r.t. a matrix $B=B^\top$ by $\lambda_{\max}(A,B)$; $I_n\in\mathbb{R}^{n\times n}$ represents the identity matrix. 

\section{Setup and Preliminaries}

We consider the discrete-time, nonlinear perturbed system
\begin{subequations}
	\label{eq:sys}
	\begin{align}
	\label{eq:sys_1}
	x_{t+1}&=f(x_t,u_t,w_t),\\
	\label{eq:sys_2}
	y_t&=h(x_t,u_t,v_t), 
	\end{align}
\end{subequations}
with state $x_t\in\mathbb{X}\subseteq\mathbb{R}^n$, control input $u_t\in\mathbb{U}\subseteq\mathbb{R}^m$, disturbances $w_t\in\mathbb{W}\subseteq\mathbb{R}^{q}$ and $v_t\in\mathbb{V}\subseteq\mathbb{R}^{r}$, noisy measurement $y_t\in\mathbb{Y}\subset\mathbb{R}^p$, and discrete time $t\in\mathbb{I}_{\geq 0}$.
The nonlinear continuous functions  $f:\mathbb{R}^n\times\mathbb{R}^m\times\mathbb{R}^q\rightarrow\mathbb{R}^n$ and  $h:\mathbb{R}^n\times\mathbb{R}^m\times\mathbb{R}^r\rightarrow\mathbb{R}^p$ constitute the system dynamics and output equation, respectively.
We assume that $0\in\mathbb{V}$ and denote the nominal (disturbance-free) system equations as $f_{\mathrm{n}}(x,u) = f(x,u,0)$ and $h_{\mathrm{n}}(x,u) = h(x,u,0)$.

Given some initial guess $\hat{x}_0$ of the true state $x_0$, the overall goal is, at any time $t\in\mathbb{I}_{\geq0}$, to provide an estimate $\hat{x}_t$ of the current state $x_t$ that satisfies the following stability notion. 

\begin{definition}[\protect{RGES \cite[Def. 1]{Knuefer2018}}]
	\label{def:RGES}
	A state estimator for system~\eqref{eq:sys} is robustly globally exponentially stable (RGES) if there exist $C_1,C_2,C_3>0$ and $\lambda_1,\lambda_2,\lambda_3\in [0,1)$ such that the resulting state estimate $\hat{x}_t$ satisfies
	\begin{align}	
	\label{eq:RGES}	
	\|x_{t}-\hat{x}_{t}&\| \leq \max \Big\{ C_1\lambda_1^t\|x_0-\hat{x}_0\|,\\	
	& \max_{j\in\mathbb{I}_{[0,t-1]}}C_2\lambda_2^{t-j-1}\|w_j\|, \max_{j\in\mathbb{I}_{[0,t-1]}}C_3\lambda_3^{t-j-1}\|v_j\|\Big\}\nonumber
	\end{align}
	for all $t\in\mathbb{I}_{\geq0}$, all initial conditions $x_0,\hat{x}_0\in\mathbb{X}$, and every trajectory $(x_t,u_t,w_t,v_t,y_t)_{t=0}^\infty$ satisfying~\eqref{eq:sys}.
\end{definition}

This corresponds to an exponential version of the robust stability property that is often used in the recent MHE literature, see, e.g., \cite{Rawlings2017,Knuefer2021,Allan2021a,Schiller2022}.

\begin{remark}\label{rem:sum_max}
	A state estimator is RGES as characterized in Definition~\ref{def:RGES} if and only if~\eqref{eq:RGES} holds with each maximization operation replaced by summation, cf. \cite[Prop.~3.13]{Allan2021a}.
\end{remark}

To establish RGES for the suboptimal MHE scheme presented in Section~\ref{sec:mhe_formulation} below, we require knowledge of an additional auxiliary observer.
To this end, we consider the following standard form given by a (possibly time-varying) observer mapping $g_t:\mathbb{Z} \times \mathbb{U} \times \mathbb{Y}\rightarrow \mathbb{Z}$ with $\mathbb{Z}\subseteq \mathbb{R}^n$ such that at any $t\in\mathbb{I}_{\geq0}$,
\begin{equation} \label{eq:obs}
z_{t+1}=g_t(z_t,u_t,y_t)
\end{equation}
is an estimate of the state $x_{t+1}$ of system~\eqref{eq:sys} using its current inputs and outputs $(u_t,y_t)$ and the estimate $z_t\in\mathbb{Z}$.
We assume that some observer in the form of~\eqref{eq:obs} is available that satisfies the following property.

\begin{assumption}[RGES observer] \label{ass:observer}
	There exists a $\delta$-Lya\-pu\-nov function $V_{\mathrm{o}} : \mathbb{Z}\times\mathbb{X}\rightarrow\mathbb{R}_{\geq0}$ and some $\eta \in [0,1)$, symmetric matrices $P_1,P_2\succ0$, and $Q,R\succeq0$ such that
	\begin{subequations}\label{eq:lyap}
		\begin{align}			
		&\ \|z-x\|_{P_1}^2 \leq V_\mathrm{o}(z,x) \leq \|z-x\|_{P_2}^2, \label{eq:lyap_1}\\
		&\ V_\mathrm{o}(g(z,u,h(x,u,v)),f(x,u,w)) \nonumber\\
		\leq &\ \eta V_{\mathrm{o}}(z,x)
		+\|w\|_Q^2 +\|v\|_R^2 \label{eq:lyap_2}
		\end{align}
	\end{subequations}
	for all $(z,x,u,w,v)\in\mathbb{Z}\times\mathbb{X}\times\mathbb{U}\times\mathbb{W}\times\mathbb{V}$.
\end{assumption}

Such a characterization of an RGES of observer was previously used in the context of MHE in \cite{Koehler2021}.
Overall, we consider a rather general class of observers in \eqref{eq:obs}, which represents an active area of research.
In particular, Assumption~\ref{ass:observer} can be verified with a quadratically bounded $\delta$-Lyapunov function $V_{\mathrm{o}}$ by employing the differential dynamics, cf. \cite{Sanfelice2016,Yi2021}.
Alternatively, we could restrict the design to a quadratic function $V_{\mathrm{o}}$, where sufficient conditions can be derived based on, e.g., incremental quadratic constraints \cite{Zhang2019} or specific Lipschitz properties \cite{Zemouche2013}; a quadratic (time-varying) function $V_{\mathrm{o}}$ arises for Kalman-like observers, cf. \cite{Reif1999,Jaganath2005}.
Note that Assumption~\ref{ass:observer} is our key assumption and can restrict the class of systems to which the proposed MHE scheme is applicable.

In Section~\ref{sec:MHE_stability}, we show that the $\delta$-Lyapunov function $V_{\mathrm{o}}$ from Assumption~\ref{ass:observer} serves as an $M$-step Lyapunov function for suboptimal MHE.
To this end, we additionally require the following continuity property of $h$~\eqref{eq:sys_2}.

\begin{assumption}[Lipschitz continuity of $h$] \label{ass:lipschitz} 
	The function $h$ is Lipschitz continuous, i.e., there exists some constant $L_{\mathrm{h}} > 0$ such that
	\begin{equation*}
	\|h(x,u,v)-h(\bar{x},\bar{u},\bar{v})\|\leq L_{\mathrm{h}}(\|x-\bar{x}\|+\|u-\bar{u}\|+\|v-\bar{v}\|)
	\end{equation*}
	for all $(x,\bar{x}) \in \mathbb{X}\times\mathbb{Z}$, $u,\bar{u} \in \mathbb{U}$, and $v,\bar{v} \in \mathbb{V}$.
\end{assumption} 

\section{Suboptimal Moving Horizon Estimation}\label{sec:mhe_formulation}
Given some finite estimation horizon $M\in\mathbb{I}_{\geq0}$, the proposed MHE scheme considers at each $t\in\mathbb{I}_{\geq0}$ the input and output data $(u,y)$ of system~\eqref{eq:sys} in a moving time window of length $M_t=\min\{t,M\}$.
The corresponding state estimate $\hat{x}_t$ is then obtained by solving the following nonlinear program (NLP)
\begin{subequations}\label{eq:MHE}
	\begin{align}\label{eq:MHE_IOSS_cost}
	\min_{\hat{x}_{t-M_t|t}}& J_t(\hat{x}_{t-M_t|t}) \\
	\text{s.t.} \ &\ \hat{x}_{j+1|t}= g_j(\hat{x}_{j|t},u_{j},y_{j}),\ j\in\mathbb{I}_{[t-M_t,t-1]}, \label{eq:MHE_1}\\
	&\ \hat{y}_{j|t}=h_{\mathrm{n}}(\hat{x}_{j|t},u_{j}),\ j\in\mathbb{I}_{[t-M_t,t-1]}, 	\label{eq:MHE_2}\\
	&\ \hat{x}_{j|t}\in\mathbb{Z},\ j\in\mathbb{I}_{[t-M_t,t]}, 	\label{eq:MHE_3}
	\end{align}
\end{subequations}
with the observer dynamics $g_t$~\eqref{eq:obs}.
Given the most recent input and output sequences $\{u_j\}_{j=t-M_t}^{t-1}$ and $\{y_j\}_{j=t-M_t}^{t-1}$ of system~\eqref{eq:sys}, the decision variable $\hat{x}_{t-M_t|t}$ (uniquely) defines a sequence of state estimates $\{\hat{x}_{j|t}\}_{j=t-M_t}^{t}$ and a sequence of output estimates $\{\hat{y}_{j|t}\}_{j=t-M_t}^{t-1}$ under \eqref{eq:MHE_1}-\eqref{eq:MHE_2}.
The (time-varying) cost function $J_t:\mathbb{R}^n\rightarrow\mathbb{R}_{\geq0}$ is defined as
\begin{align}
J_t(\hat{x}_{t-M_t|t}) =&\ 2\|\hat{x}_{t-M_t|t}-\hat{x}_{t-M_t}\|_{W}^2 \label{eq:MHE_objective} \\
&\ + \frac{\lambda_{\min}(P_1)}{2L_{\mathrm{h}}^2\lambda_{\max}(G)} \sum_{j=1}^{M_t}\eta^{j}\|\hat{y}_{t-j|t}-y_{t-j}\|_{G}^2, \nonumber
\end{align}
where $\hat{x}_{t-M_t}$ is the MHE estimate obtained $M_t$ steps in the past, $\eta$ and $P_1$ are from the $\delta$-Lyapunov function $V_{\mathrm{o}}$ (Assumption~\ref{ass:observer}), and $L_{\mathrm{h}}$ is from the Lipschitz property of~$h$ (Assumption~\ref{ass:lipschitz}).
The parameters $W,G\succeq0$ with $W,G\neq0$ are weighting matrices that can be tuned arbitrarily; their respective influence on the theoretical properties of the resulting estimator is discussed in detail in Remark~\ref{rem:param} below.

Note that in contrast to most of the literature on nonlinear MHE, the proposed scheme does not  optimize over a disturbance sequence $\{\hat{w}_{j|t}\}_{j=t-M_t}^t$ as is the case in, e.g., \cite{Schiller2022,Knuefer2021,Allan2021a,Schiller2022b,Rawlings2017}; 
instead, we directly employ the observer dynamics~\eqref{eq:obs} in~\eqref{eq:MHE_1}.
This direct coupling between the MHE formulation and the auxiliary observer allows for using the corresponding $\delta$-Lyapunov function $V_{\mathrm{o}}$ as $M$-step Lyapunov function for (suboptimal) MHE, cf.~Section~\ref{sec:MHE_stability} below.
As a consequence, the estimated states $\hat{x}_{j|t}$ are restricted to the set $\mathbb{Z}$ via constraint~\eqref{eq:MHE_3}, i.e., to the domain of the observer~\eqref{eq:obs}.
If additional knowledge on the real system state is available (e.g., due to physical nature) and should be incorporated into the MHE scheme to improve performance, one could suitably re-design the auxiliary observer as suggested in \cite{Astolfi2021} or use additional projections as in \cite[Sec.~VI]{Schiller2022b}.

We consider the following suboptimal estimator.

\begin{definition}[Suboptimal estimator] \label{def:estimator}
	Let $t \in \mathbb{I}_{\geq 0}$, $M\in \mathbb{I}_{\geq1}$, the past estimate $\hat{x}_{t-M_t}$, and the input/output sequence $(u_j,y_j)_{j={t-M_t}}^{t-1}$ of system~\eqref{eq:sys} be given and let $\tilde{x}_{t-M_t|t} \in \mathbb{Z}$ denote a feasible candidate solution to the MHE problem~\eqref{eq:MHE}.
	Then, the corresponding suboptimal solution of~\eqref{eq:MHE} is defined as any estimate~$\hat{x}_{t-M_t|t}\in\mathbb{Z}$ such that 1) the MHE constraints \eqref{eq:MHE_1}-\eqref{eq:MHE_3} and 2) the cost decrease condition
	\begin{equation}
	J_t(\hat{x}_{t-M_t|t}) \leq 
	J_t(\tilde{x}_{t-M_t|t}) \label{eq:cost_decrease}
	\end{equation}
	are satisfied. The (suboptimal) state estimate at time $t \in \mathbb{I}_{\geq 0}$ is defined as $\hat{x}_t = \hat{x}_{t|t}$.
\end{definition}

We consider the following choice of the candidate solution
\begin{equation}
\tilde{x}_{t-M_t|t} = \hat{x}_{t-M_t}. \label{eq:candidate}
\end{equation}
This candidate solution\footnote{
	Note that~\eqref{eq:candidate} does not restrict the warm start of the particular algorithm used to solve~\eqref{eq:MHE}; a practical choice is, e.g., $\hat{x}_{t-M_t|t-1}$, i.e., the most recent MHE estimate, compare also~\cite[Rem.~4]{Schiller2022b}.
} is much simpler in contrast to our recent result \cite{Schiller2022b}, in which the auxiliary observer needed to be re-initialized, re-simulated, and transformed into a trajectory of system~\eqref{eq:sys}.
In the next section, we derive practical conditions for robust stablity of suboptimal MHE, simply by exploiting the coupling between the MHE problem~\eqref{eq:MHE} and the RGES observer~\eqref{eq:obs} satisfying Assumption~\ref{ass:observer}.

\section[M-Step Lyapunov Function for Suboptimal MHE]{$M$-Step Lyapunov Function for Suboptimal~MHE}
\label{sec:MHE_stability}
In order to show that $V_{\mathrm{o}}$ is an $M$-step Lyapunov function for suboptimal MHE, we require the following auxiliary lemma that provides a bound on the cost of the candidate solution $\tilde{J}_t := J(\tilde{x}_{t-M_t|t})$.
\begin{lemma}\label{lem:candidate_cost}
	Let Assumptions~\ref{ass:observer} and \ref{ass:lipschitz} hold and $\hat{x}_0\in\mathbb{Z}$.
	Then, the cost function~\eqref{eq:MHE_objective} evaluated at the candidate solution~\eqref{eq:candidate} satisfies for all $t\in\mathbb{I}_{\geq0}$
	\begin{align}
	\tilde{J}_t
	\leq &\ M_t\eta^{M_t}V_{\mathrm{o}}(\hat{x}_{t-M_t},x_{t-M_t}) + M_t \sum_{j=1}^{M_t}\eta^{j-1} \|w_{t-j}\|_Q^2 \nonumber\\
	&\ + \left(\eta\frac{\lambda_{\min}(P_1)}{\lambda_{\min}(R)}+M_t\right)  \sum_{j=1}^{M_t}\eta^{j-1} \|v_{t-j}\|_R^2. \label{eq:lem_res}
	\end{align}
\end{lemma}

\begin{proof}
	First, note that the candidate solution~\eqref{eq:candidate} represents a feasible choice, which follows from the definition of the observer~\eqref{eq:obs} and constraint~\eqref{eq:MHE_3} since $\hat{x}_0\in\mathbb{Z}$; the corresponding state and output sequences generated from \eqref{eq:MHE_1}-\eqref{eq:MHE_2} are denoted as $\{\tilde{x}_{j|t}\}_{j=t-M_t}^{t}$ and $\{\tilde{y}_{j|t}\}_{j=t-M_t}^{t-1}$, respectively.
	From the cost function~\eqref{eq:MHE_objective}, we have that
	\begin{align}
	\tilde{J}_t \leq \frac{\lambda_{\min}(P_1)}{2L_{\mathrm{h}}^2\lambda_{\max}(G)}\sum_{j=1}^{M_t}\eta^{j}\|\tilde{y}_{t-j|t}-y_{t-j}\|_{G}^2
	\label{eq:proof_lem_cost}
	\end{align}
	since $\tilde{x}_{t-M_t|t} =\hat{x}_{t-M_t}$ due to~\eqref{eq:candidate}.
	Using Assumption~\ref{ass:lipschitz} (and, for brevity, omitting indices in the following step) leads to
	\begin{align}
	&\ \|\tilde{y}-y\|_G^2 \leq \lambda_{\max}(G)\|h(\tilde{x},u,0)-h(x,u,v)\|^2 \nonumber \\
	\leq&\ \lambda_{\max}(G)2L_{\mathrm{h}}^2(\|\tilde{x}-x\|^2+\|v\|^2) \nonumber \\
	\leq&\ 2L_{\mathrm{h}}^2\left( \frac{\lambda_{\max}(G)}{\lambda_{\min}(P_1)}\|\tilde{x}-x\|_{P_1}^2 +  \frac{\lambda_{\max}(G)}{\lambda_{\min}(R)}\|v\|_R^2\right). \label{eq:proof_lem_xw}
	\end{align}	
	By combining~\eqref{eq:proof_lem_cost}, \eqref{eq:proof_lem_xw}, and~\eqref{eq:lyap_1}, we therefore obtain
	\begin{equation}
	\tilde{J}_t
	\leq \sum_{j=1}^{M_t}\eta^{j}\left( V_{\mathrm{o}}(\tilde{x}_{t-j|t},x_{t-j}) + \frac{\lambda_{\min}(P_1)}{\lambda_{\min}(R)}\|v_{t-j}\|_R^2 \right)\label{eq:proof_lem_VQ}
	\end{equation}
	and, since $\{\tilde{x}_{j|t}\}_{j=t-M_t}^{t}$ is a state trajectory of the observer~\eqref{eq:obs} via~\eqref{eq:MHE_3}, we can invoke Assumption~\ref{ass:observer} and apply the dissipation inequality~\eqref{eq:lyap_2} for each $j\in\mathbb{I}_{[1,M_t]}$ $(M_t-j)$ times. This leads to
	\begin{align*}
	&\ \eta^{j}V_{\mathrm{o}}(\tilde{x}_{t-j|t},x_{t-j}) \leq \eta^j\Bigg(\eta^{M_t-j}V_o(\tilde{x}_{t-M_t|t},x_{t-M_t}) \\
	&\ + \sum_{i=j+1}^{M_t}\eta^{i-j-1}\big(\|w_{t-i}\|_Q^2 + \|v_{t-i}\|_R^2\big)\Bigg)
	\end{align*}
	for each $j\in\mathbb{I}_{[1,M_t]}$.
	Summing up over all $j\in\mathbb{I}_{[1,M_t]}$ yields
	\begin{align}
	&\ \sum_{j=1}^{M_t}\eta^{j}V_{\mathrm{o}}(\tilde{x}_{t-j|t},x_{t-j}) \leq  M_t\eta^{M_t}V_{\mathrm{o}}(\tilde{x}_{t-M_t|},x_{t-M_t}) \nonumber \\
	&\qquad + M_t\sum_{j=1}^{M_t}\eta^{j-1}\left(\|w_{t-j}\|_Q^2+\|v_{t-j}\|_R^2\right). \label{eq:proof_lem_sum}
	\end{align}
	Combining \eqref{eq:proof_lem_VQ} and \eqref{eq:proof_lem_sum} and recalling that $\tilde{x}_{t-M_t|t} = \hat{x}_{t-M_t}$ by~\eqref{eq:candidate} leads to \eqref{eq:lem_res}, which hence concludes this proof.
\end{proof}

In the following, we show that $V_{\mathrm{o}}$ is an $M$-step Lyapunov function for suboptimal MHE.

\begin{theorem}\label{thm}
	Let Assumptions~\ref{ass:observer} and \ref{ass:lipschitz} hold and $\hat{x}_0\in\mathbb{Z}$.
	Then, the suboptimal state estimate $\hat{x}_{t}$ satisfies
	\begin{align}
	V_{\mathrm{o}}(\hat{x}_t,x_t) &\leq \gamma_1(M_t)V_{\mathrm{o}}(\hat{x}_{t-M_t},x_{t-M_t})\label{eq:thm_result}\\
	& \, + \sum_{j=1}^{M_t}\eta^{j-1}(\gamma_2(M_t)\|w_{t-j}\|_Q^2+\gamma_3(M_t)\|v_{t-j}\|_R^2)	\nonumber
	\end{align}
	for all $t\in\mathbb{I}_{\geq0}$, where
	\begin{subequations}\label{eq:gamma}
		\begin{align}
		&\overline{\gamma}_1(k,r,s){\,:=} 2\lambda_{\max}(P_2,P_1)\eta^{s}{+}\lambda_{\max}(P_2,W)k\eta^{r+s},\label{eq:gamma_1}\\
		&\overline{\gamma}_2(k,r){\,:=} 1+\lambda_{\max}(P_2,W)k\eta^{r},\label{eq:gamma_2}\\
		&\overline{\gamma}_3(k,r){\,:=} 1+\lambda_{\max}(P_2,W)\left(\eta\frac{\lambda_{\min}(P_1)}{\lambda_{\min}(R)}+k\right)\eta^{r},\label{eq:gamma_3}
		\end{align}
	\end{subequations}
	with\footnote{
		We define the functions $\overline{\gamma}_i$ in~\eqref{eq:gamma} as functions of three (two) separate variables, since this will be convenient for various extensions/adaptations discussed in Remark~\ref{rem:alternative_candidate} and Section~\ref{sec:num}.
	} $\gamma_1(r):= \overline{\gamma}_1(r,r,r)$, and $\gamma_i(r):= \overline{\gamma}_i(r,r)$, $i = \{2,3\}$.
\end{theorem}

\begin{proof}
	At any $t\in\mathbb{I}_{\geq0}$, constraint~\eqref{eq:MHE_1} ensures that the estimated suboptimal trajectory $\{\hat{x}_{j|t}\}_{j=t-M_t}^{t}$ is a trajectory of the observer~\eqref{eq:obs}, which by Assumption~\ref{ass:observer} admits the $\delta$-Lyapunov function $V_{\mathrm{o}}(\tilde{x},x)$.
	Hence, we can apply the dissipation inequality~\eqref{eq:lyap_2} for $M_t$ times, which leads to
	\begin{align}
	V_{\mathrm{o}}(\hat{x}_t,x_t)\leq&\ \eta^{M_t}V_{\mathrm{o}}(\hat{x}_{t-M_t|t},x_{t-M_t}) \label{eq:proof_thm_lyap} \\
	&\ + \sum_{j=1}^{M_t}\eta^{j-1}\big(\|w_{t-j}\|_Q^2+\|v_{t-j}\|_R^2\big). \nonumber
	\end{align}
	Using \eqref{eq:lyap_1} with Cauchy-Schwarz and Young’s inequality, we further have that
	\begin{align}
	&\ V_{\mathrm{o}}(\hat{x}_{t-M_t|t},x_{t-M_t}) \leq  \|\hat{x}_{t-M_t|t}-x_{t-M_t}\|_{P_2}^2 \nonumber\\
	\leq&\ 2\|\hat{x}_{t-M_t|t}-\hat{x}_{t-M_t}\|_{P_2}^2+2\|x_{t-M_t}-\hat{x}_{t-M_t}\|_{P_2}^2. \label{eq:proof_thm_triangle}
	\end{align}
	The second term of the right-hand side in \eqref{eq:proof_thm_triangle} can be bounded by exploiting \eqref{eq:lyap_1} according to
	\begin{equation}
		2\|\hat{x}_{t-M_t}-x_{t-M_t}\|_{P_2}^2
		\leq 2\lambda_{\max}(P_2,P_1)V_{\mathrm{o}}(\hat{x}_{t-M_t},x_{t-M_t}). \label{eq:proof_thm_term2}
	\end{equation}
	Using a similar reasoning, the first term of the right-hand side in \eqref{eq:proof_thm_triangle} satisfies
	\begin{align}
	2\|\hat{x}_{t-M_t|t}-\hat{x}_{t-M_t}\|_{P_2}^2
	\leq \lambda_{\max}(P_2,W)J_t(\hat{x}_{t-M_t|t}), \label{eq:proof_thm_term1}
	\end{align}
	which follows from the definition (and non-negativity) of the cost function~\eqref{eq:MHE_objective}.
	Now recall that  $J_t(\hat{x}_{t-M_t|t})\leq \tilde{J}_t$ due to~\eqref{eq:cost_decrease};
	consequently, we can invoke Lemma~\ref{lem:candidate_cost}, and thus, the combination of \eqref{eq:proof_thm_lyap} and \eqref{eq:proof_thm_triangle}-\eqref{eq:proof_thm_term1} leads to~\eqref{eq:thm_result}, which hence concludes this proof.
\end{proof}

Provided that $M\in\mathbb{I}_{\geq1}$ is chosen such that
\begin{equation}
\rho^M:= \gamma_1(M)<1 \label{eq:cond_M}
\end{equation}
holds, Theorem~\ref{thm} directly implies that
\begin{align}
V_{\mathrm{o}}(\hat{x}_t,x_t) &\leq \rho^MV_{\mathrm{o}}(\hat{x}_{t-M_t},x_{t-M_t})  \nonumber\\
&\ \, + \sum_{j=1}^{M}\eta^{j-1}\left(\gamma_2(M)\|w_{t-j}\|_Q^2 + \gamma_3(M)\|v_{t-j}\|_R^2\right) \nonumber
\end{align}
for $t\in\mathbb{I}_{\geq M}$. Consequently, $V_{\mathrm{o}}$ is an $M$-step Lyapunov function for suboptimal MHE, compare~\cite[Thm.~1]{Schiller2022} for standard MHE (without auxiliary observer). 
We can straightforwardly deduce RGES as shown in the following corollary.

\begin{corollary}
	Let Assumptions~\ref{ass:observer} and \ref{ass:lipschitz} hold, $\hat{x}_0\in\mathbb{Z}$,~and $M\in\mathbb{I}_{\geq0}$ satisfy \eqref{eq:cond_M}.
	Then, the suboptimal moving horizon estimator from Definition~\ref{def:estimator} is RGES.
\end{corollary}
\begin{proof}
	The proof is straightforward: applying standard Lyapunov arguments to~\eqref{eq:thm_result} under~\eqref{eq:cond_M} and exploiting Remark~\ref{rem:sum_max} yields the desired result.
\end{proof}

Some remarks are in order.

\begin{remark}[Condition on the horizon length]\label{rem:cond_M}
	By standard properties of the exponential function, one can easily verify specific properties of $\gamma_1:\mathbb{R}_{\geq0}\rightarrow\mathbb{R}_{\geq0}$ on the open interval $[0,\infty)$ --- namely, that $\gamma_1$ is continuous, has one (global) maximum, and $\lim_{M\rightarrow\infty}\gamma_1(M)=0$.
	Consequently, there exists some $\underline{M}\in\mathbb{I}_{\geq1}$ such that \eqref{eq:cond_M} holds for all $M\in\mathbb{I}_{\geq \underline{M}}$.
	In practice, a sufficiently large $\underline{M}$ can be easily obtained by solving~\eqref{eq:cond_M} numerically.
\end{remark}

\begin{remark}[Parameterization of the cost function]\label{rem:param}
	The matrices $W$ and $G$ in~\eqref{eq:MHE_objective} are arbitrary tuning parameters.
	The choice of $G$ has no impact on the theoretical guarantees (note that $G$ does not appear in \eqref{eq:thm_result}), since the stage cost is normalized by its largest eigenvalue $\lambda_{\max}(G)$. Consequently, $G$ can be used to scale the output estimates differently in case $p>1$.
	In contrast, $W$ has a direct influence on all functions $\gamma_1,\gamma_2,\gamma_3$~\eqref{eq:gamma} via the generalized eigenvalue $\lambda_{\max}(P_2,W)$.
	This can be exploited to adjust the degree of confidence in the observer's estimates by specifying how much the estimated trajectory $\{\hat{x}_{j|t}\}_{t-M}^t$ may ($\lambda_{\max}(P_2,W)\gg1$) or may not ($\lambda_{\max}(P_2,W)\ll1$) deviate from the observer trajectory initialized at $\hat{x}_{t-M_t}$.
	For small values of $\lambda_{\max}(P_2,W)$, the minimum horizon length is dominated by the first factor in \eqref{eq:gamma_1} and the functions $\gamma_2,\gamma_3$ in \eqref{eq:gamma_2}-\eqref{eq:gamma_3} become closer to that of the observer, cf.~\eqref{eq:lyap_2}.
	Conversely, the further one deviates from the stabilizing observer by choosing large values of $\lambda_{\max}(P_2,W)$ in~\eqref{eq:gamma}, the worse the guarantees become and the larger the horizons must be chosen; on the other hand, this choice typically leads to good results in practice, since the estimate from the auxiliary observer can (potentially significantly) be improved with only a few iterations, compare also the simulation example in Section~\ref{sec:num}.
\end{remark}

\begin{remark}[Asymptotic behavior for large $M$]\label{rem:M_large}
	Similar properties as discussed in Remark~\ref{rem:cond_M} for the function $\gamma_1$ also apply to $\gamma_2$ and $\gamma_3$. In particular, both these functions are monotonically decreasing in $M$ for $M$ large enough, and $\lim_{M\rightarrow\infty}\gamma_2(M)=\lim_{M\rightarrow\infty}\gamma_3(M) = 1$.
	Together with $\lim_{M\rightarrow\infty}\gamma_1(M) = 0$ (cf. Remark~\ref{rem:cond_M}), this implies the appealing theoretical feature that for $M\rightarrow\infty$, the bound from Theorem~\ref{thm} converges to the bound given by the $\delta$-Lyapunov function $V_{\mathrm{o}}$~\eqref{eq:lyap}, regardless of how the cost function~\eqref{eq:MHE_objective} is parameterized.
	This is generally not the case in \cite{Schiller2022b}, where the guarantees for suboptimal MHE are strictly worse than those from the auxiliary observer.
\end{remark}

\begin{remark}[Alternative Lyapunov function]
	The recent result \cite{Schiller2022} uses a $\delta$-IOSS Lyapunov function (which characterizes the detectability of the system) as $M$-step Lyapunov function for (standard) MHE.
	In contrast, we use the \mbox{$\delta$-Lyapunov} function $V_{\mathrm{o}}$ (which characterizes robust stability of the auxiliary observer) as $M$-step Lyapunov function for suboptimal MHE.
	A natural alternative could be to use the $\delta$-IOSS Lyapunov function also for suboptimal MHE, which becomes possible by combining the new Lyapunov approach from \cite{Schiller2022} with the suboptimal MHE scheme from \cite{Schiller2022b}.
	This results in the more standard MHE formulation where \eqref{eq:MHE_1}-\eqref{eq:MHE_3} is replaced according to the system dynamics~\eqref{eq:sys} and one additionally optimizes over a disturbance and noise sequence, cf.~\cite{Schiller2022}.
	However, due to the mismatch between the MHE constraints and the dynamics of the auxiliary observer~\eqref{eq:obs} used to construct the candidate solution, the overall guarantees that can be established for suboptimal MHE become more conservative in this case.
\end{remark}

\begin{remark}[Alternative candidate solution]\label{rem:alternative_candidate}
	For $M\in\mathbb{I}_{\geq1}$ arbitrarily fixed, we could also apply the re-initialization procedure that was suggested in~\cite{Schiller2022b} and derive a $T$-step Lyapunov function for a sufficiently large $T\in\mathbb{I}_{\geq M}$, thus ensuring robust stability of suboptimal MHE for an arbitrary horizon length $M$.
	However, the candidate solution becomes more intricate. In particular, at each $t\in\mathbb{I}_{\geq0}$, we need to re-initialize the auxiliary observer $T_t:=\min\{t,T\}$ steps in the past using $z_{t-T_t|t} = \hat{x}_{t-T_t}$ and perform a forward simulation for $T_t-M_t$ steps to obtain the candidate solution $\tilde{x}_{t-M_t|t} = z_{t-M_t|t}$; in addition, $\hat{x}_{t-M_t}$ needs to be replaced by $z_{t-M_t|t}$ in~\eqref{eq:MHE_objective}, compare~\cite{Schiller2022b}.
	Then, by suitably modifying the proofs of Lemma~\ref{lem:candidate_cost} and Theorem~\ref{thm}, we can derive~\eqref{eq:thm_result}-\eqref{eq:gamma} with functions $\overline{\gamma}_1(M_t,M_t,T_t)$, ${\gamma}_2(M_t)$, ${\gamma}_3(M_t)$, where $1$ is replaced by $2\lambda_{\max}(P_2,P_1)$ in \eqref{eq:gamma_2}-\eqref{eq:gamma_3}.
	Condition~\eqref{eq:cond_M} (with $\gamma_1(M)$ replaced by $\overline{\gamma}_1(M,M,T)$) can then be easily solved for a sufficient $T$, compare the example in Section~\ref{sec:num}.
\end{remark}

\section{Numerical Example}
\label{sec:num} 
We consider the following system 
\begin{subequations}
	\label{eq:num_sys}
	\begin{align}
	x_1^+ &= x_1 + t_{\Delta}(-2k_1x_1^2+2k_2x_2) + w_1,\\
	x_2^+ &= x_2 + t_{\Delta}(k_1x_1^2-k_2x_2) + w_2,\\
	y &= x_1 + x_2+ v,
	\end{align}
\end{subequations}
with $k_1=0.16$, $k_2=0.0064$, which corresponds to the discretized nonlinear chemical reactor process from~\cite[Sec.~5]{Tenny2002} using the sampling time $t_{\Delta}=0.1$.
This constitutes a standard benchmark example in the context of nonlinear MHE, and as is usually done in this setting, we choose $x_0=[3,1]^\top$ and the poor initial estimate $\hat{x}_0=[0.1,4.5]^\top$.
The disturbances $w$ and $v$ are modeled as uniformly distributed random variables sampled from $\mathbb{W} = \lbrace w \in \mathbb{R}^{2} : |w_i| \leq 2\cdot10^{-3}, i = \lbrace 1, 2 \rbrace \rbrace$ and $\mathbb{V} = \lbrace v \in \mathbb{R} : |v| \leq 10^{-2} \rbrace$ during the simulation.

For system~\eqref{eq:num_sys}, we design a Luenberger observer with $g(z,y)=f_{\mathrm{n}}(z)+L(h_{\mathrm{n}}(z)-y)$ in~\eqref{eq:obs}.
The constant observer gain $L\in\mathbb{R}^{n\times p}$ is computed based on the differential dynamics, where a sufficient linear matrix inequality (LMI) analogous to the dual (i.e., control) problem \cite{Manchester2018} can be derived. 
By imposing a quadratic Lyapunov function $V_{\mathrm{o}}(z,x)= \| z-x\|_P^2$, we can verify\footnote{
	LMIs were solved in Matlab using the toolbox YALMIP \cite{Loefberg2004} and the semidefinit programm solver MOSEK \cite{MOSEKApS2019}.
} Assumption~\ref{ass:observer} on $\mathbb{Z}=\{z\in\mathbb{R}^n : 0.1\leq z_1 \leq 6\}$  with gain $L=[7.999,-9.997]^\top$, $P = ${\footnotesize$\begin{bmatrix}  1.537 & 1.380 \\ 1.380 & 1.254\end{bmatrix}$}, $\eta=0.955$, $Q = 10^3\cdot I_2$, and $R=100$.

\subsubsection*{Proposed Suboptimal MHE}
We choose $G=1$ and $W=aP_2$ with $a>0$ in~\eqref{eq:MHE_objective}, which implies that $\lambda_{\max}(P_2,W)=1/a$; in the first simulation, we consider two different values of $a$ to illustrate Remark~\ref{rem:param}.
As a performance benchmark, we additionally consider the standard (i.e., fully optimized w.r.t. the system dynamics) MHE formulation from \cite{Schiller2022} and parameterize its cost function for the sake of comparability by verifying \cite[Cor.~2]{Schiller2022} with $V_{\mathrm{o}}$ as $\delta$-IOSS Lyapunov\footnote{
	Note that using a different (worse conditioned) matrix $P$ allows for choosing a smaller horizon length, compare \cite[Sec.~V]{Schiller2022}.
} function.
We implement each estimator in the filtering\footnote{
	The results derived in Section~\ref{sec:MHE_stability} (and \cite{Schiller2022}) for the prediction form of MHE (i.e., neglecting the current measurement $y_t$) can be easily extended to the filtering form of MHE (i.e., including $y_t$), albeit with a (significantly) more cumbersome notation, compare also \cite[Sec.~4.2]{Rawlings2017}; this yields \eqref{eq:thm_result}-\eqref{eq:gamma} with $\overline{\gamma}_1(M_t+1,M_t,M_t)$, $\overline{\gamma}_2(M_t+1,M_t)$, and $\overline{\gamma}_3(M_t+1,M_t)/\eta$. \label{footnote}
} form of MHE, since this is generally beneficial in practice.
The respective horizon lengths $\underline{M}$ ensuring condition~\eqref{eq:cond_M} are shown in the upper part of Table~\ref{tab:comparison_1}.
\begin{table}
	\begin{threeparttable}
		\caption{Performance of the proposed MHE scheme compared to \cite{Schiller2022}.}
		\label{tab:comparison_1}
		\setlength\tabcolsep{4pt} 
		\begin{tabular*}{\columnwidth}{@{\extracolsep{\fill}}c ccc ccc c}
			\toprule
			Setup
			& \multicolumn{3}{c}{$a = 10^2$,  $\underline{M}{=}16$}
			& \multicolumn{3}{c}{$a = 10^{-3}$, $\underline{M}{=}128$}
			& \cite{Schiller2022}, $\underline{M}{=}30$\\
			\cmidrule(lr){1-1}\cmidrule(lr){2-4}\cmidrule(lr){5-7}\cmidrule(lr){8-8}
			iter. $i$
			& 0 & 1 & $\ast$
			& 0 & 1 & $\ast$
			& $\ast$\\
			SSE
			& 42.87 & 42.86 & 42.85
			& 42.94 &  3.48 & 3.47
			& 0.67 \\
			\ $\tau_{\mathrm{max}}$ {\scriptsize[$\mathrm{ms}$]}
			& 3.70 & 4.72 & 4.86
			& 3.77 & 5.00 & 7.41
			& 14.59 \\\bottomrule
		\end{tabular*}
		\begin{tablenotes}
			\footnotesize
			\item Average values over 100 simulations of length $t_{\mathrm{sim}}=200$; asterisks represent fully converged optimization problems.
		\end{tablenotes}
	\end{threeparttable}
\end{table}
We simulate\footnote{
	The simulations were performed on a standard PC (Intel Core i7 with 2.6 GHz, 12 MB cache, and 16 GB RAM under Ubuntu Linux 20.04) in Matlab with CasADi \cite{Andersson2018} and the NLP solver IPOPT \cite{Waechter2005}.
} each estimator using $M=\underline{M}$, so that valid theoretical guarantees are obtained in each case.
Table~\ref{tab:comparison_1} shows the sum-of-squares error $\mathrm{SSE}=\sum_{t=0}^{t_{\mathrm{sim}}}\|\hat{x}_t-x_t\|^2$ and the maximum computation time $\tau_{\mathrm{max}}$ per sample for different numbers of iterations~$i$.
Here, we require small values of $a$ to improve the estimates from the Luenberger observer.
In line with Remark~\ref{rem:param}, this requires larger horizons in order to satisfy condition~\eqref{eq:cond_M}.
However, we find that already $i=1$ iteration is sufficient to significantly improve the estimates of the auxiliary observer ($i=0$), reducing the computational complexity (i.e., $\tau_{\mathrm{max}}$) compared to standard MHE by $\approx66\, \%$.

\subsubsection*{Comparison with the Literature}
We compare the proposed MHE framework to suboptimal MHE from \cite{Schiller2022b} and to the fast MHE schemes from \cite{Kuehl2011} and \cite{Wynn2014}.
For comparison reasons, we fix the horizon length to $M=3$, and thus, consider the modifications from Remark~\ref{rem:alternative_candidate}.
Motivated by the findings from Table~\ref{tab:comparison_1}, we choose $W=10^{-3}P_2$, and correspondingly, $T{\,=\,}\underline{T}{\,=\,}178$, which ensures satisfaction of~\eqref{eq:cond_M} with $\gamma_1$ replaced by $\overline{\gamma}_1(M+1,M,T)$, compare Remark~\ref{rem:alternative_candidate} and Footnote~\ref{footnote}.
For suboptimal MHE \cite{Schiller2022b}, we use the observer-based candidate solution with a quadratic cost function based on the $\delta$-IOSS Lyapunov function $W_\delta$, which yields $\underline{T}=302$, cf. \cite[Thm.~3]{Schiller2022b}; the cost function used in \cite{Kuehl2011} is parameterized analogously (the framework from~\cite{Wynn2014} does not provide tuning possibilities).
Since \cite{Kuehl2011} and \cite{Wynn2014} both rely on the generalized Gauß-Newton (GGN) algorithm, we implement the suboptimal MHE schemes in a similar fashion.
To this end, we must change the initial estimate to $\hat{x}=[2.3, 1.5]^\top$, since the GGN algorithm does not converge using the setup from above, illustrating its local nature, compare~\cite{Wynn2014}.
As additional benchmark, we again consider the standard MHE scheme from \cite{Schiller2022}, although the corresponding guarantees do not hold anymore since $M=3<\underline{M}=30$.

Figure \ref{fig:comparison} shows the estimation error in Lyapunov coordinates for all estimators under study, which reveals a slightly improved transient behavior of the proposed suboptimal MHE scheme compared to those from \cite{Schiller2022b,Kuehl2011,Wynn2014}, and the auxiliary observer.
\begin{figure}
	\vspace{1.8mm}
	\includegraphics{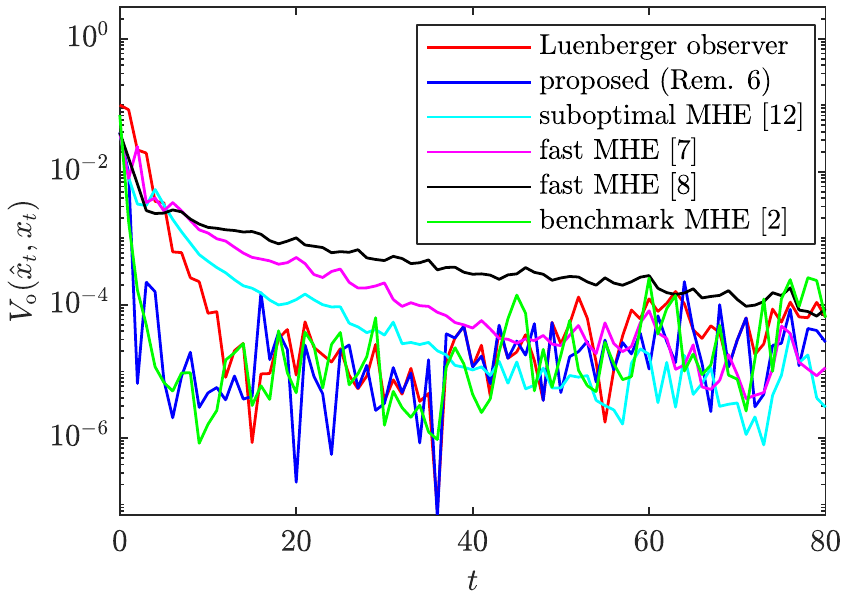}
	\caption{Estimation error in Lyapunov coordinates for the proposed suboptimal MHE scheme (with modifications from Remark~\ref{rem:alternative_candidate}) and comparable methods from the literature (suboptimal MHE~\cite{Schiller2022b}, fast MHE~\cite{Kuehl2011,Wynn2014}) after performing $i=1$ GGN iteration compared to the Luenberger observer and the (fully optimized) MHE \cite{Schiller2022} using a fixed estimation horizon $M=3$.}
	\label{fig:comparison}
\end{figure}
From the computation times shown in Table~\ref{tab:comparison_2}, we observe that the proposed scheme is slightly faster than suboptimal MHE~\cite{Schiller2022b}, which mainly is due to the fact that less decision variables are used in the optimization problem~\eqref{eq:MHE}.
\begin{table}
	\begin{threeparttable}
		\caption{Maximum computation time $\tau_{\mathrm{max}}$ per sample for the proposed suboptimal scheme compared to similar methods from the literature and \cite{Schiller2022}.}
		\label{tab:comparison_2}
		\setlength\tabcolsep{4pt} 
		\begin{tabular*}{\columnwidth}{@{\extracolsep{\fill}} cc cc cc}
			\toprule
			\ MHE scheme
			& proposed (Rem.~\ref{rem:alternative_candidate})
			& \cite{Schiller2022b}
			& \cite{Kuehl2011}
			& \cite{Wynn2014}
			& \cite{Schiller2022}\\
			$\tau_{\mathrm{max}}$ {\scriptsize[$\mathrm{ms}$]}
			& 2.39 & 2.83 & 1.08 & 0.90 & 14.31 \\
			\bottomrule
		\end{tabular*}
		\begin{tablenotes}
			\footnotesize
			\item Average values over 100 simulations of length $t_{\mathrm{sim}} = 400$.
		\end{tablenotes}
	\end{threeparttable}
\end{table}
Second, the proposed scheme (and suboptimal MHE~\cite{Schiller2022b}) is slower than the fast MHE schemes, since the forward simulation of the auxiliary observer becomes computationally more significant for large $T$, cf.~Remark~\ref{rem:alternative_candidate}.
Overall, we note that the proposed framework reduces the overall computation time per sample compared to standard MHE~\cite{Schiller2022} by ${\approx83}\, \%$.
In addition, the proposed framework is generally more flexible compared to \cite{Kuehl2011} and \cite{Wynn2014} (in particular, since arbitrary optimization algorithms can be used), and achieves good performance both in terms of accuracy and computation time under valid theoretical guarantees.

\section{Conclusion}
We presented a simple suboptimal MHE framework and provided practical sufficient conditions for guaranteed robust stability.
Given an auxiliary observer that admits a Lyapunov function, we showed that this function directly serves as $M$-step Lyapunov function for suboptimal MHE if $M$ is suitably chosen.
The derived guarantees are independent of the optimization algorithm, hold for an arbitrary number of solver iterations, improve as $M$ increases, and asymptotically approach those from the auxiliary observer, i.e., the theoretically best possible result.
The simulation example revealed that with only one iteration of the optimizer, we were able to achieve good performance in terms of both accuracy and computation time under valid theoretical guarantees.


\end{document}